\documentclass[12pt,a4paper]{article}
\usepackage{setspace} 
\usepackage{amsfonts, enumerate, amsmath, amsthm, mathrsfs}

\usepackage{natbib}
\setcitestyle{authoryear,open={(},close={)}} 

\newtheorem{theorem}{Theorem}
\newtheorem{corollary}{Corollary}
\newtheorem{lemma}{Lemma}
\newtheorem{definition}{Definition}

\def\d{{\rm d}}
\def\E{{\rm E}}

\def\P{{\rm P}}

\newcommand{\var}{\textrm{Var}}
\newcommand{\ind}{\mathbb{I}}
\newcommand{\BR}{\mathbb{R}}

\newcommand{\mF}{\mathscr{F}}

\usepackage{graphicx, color}

\usepackage{mathtools}

\begin{document}
	
\title{A New Look at Fiducial Inference}
\author{Pier Giovanni Bissiri\footnote{Department of Economics, Management and Statistics, University of Milano-Bicocca, Italy. email:pier.bissiri@unimib.it}, Chris Holmes\footnote{Department of Statistics, University of Oxford, UK. email:chris.holmes@stats.ox.ac.uk}, Stephen Walker\footnote{Department of Mathematics, University of Texas at Austin, USA. email: s.g.walker@math.utexas.edu}}
\date{}

\maketitle

\begin{abstract}
Since the idea of fiducial inference was put forward by Fisher, researchers have been attempting to place it within a rigorous and well motivated framework. It is fair to say that a general definition has remained elusive. 
In this paper we start with a representation of Bayesian posterior distributions provided by Doob that relies on martingales. This is explicit in defining how a true parameter value should depend on a random sample and hence an approach to ``inverse probability'' as originally conceived by Fisher. Taking this as our cue, we introduce a definition of fiducial inference that can be regarded as general.
\end{abstract}

Keywords:	Generalized fiducial inference; Inverse probability; Martingales; Uncertainty quantification.

\section{Introduction}
Fiducial inference was founded on an idea by Fisher regarding inverse probability \citep{Fisher1930}  for producing a Bayesian style of posterior distribution that was motivated and driven solely by the likelihood function of a statistic. A comprehensive history and the state of the fiducial argument at the time is provided by \cite{Zabell1992}, with more recent commentaries available in \cite{zabell2022} and \cite{dawid2024}. 

Our aim is not to reproduce another history; our starting point is the paper of \cite{Hannig2009}, also \cite{Hannig2016}, which introduces a general version of fiducial inference.  
Hannig's 2009 paper gives a principle under which a true parameter value can be obtained from a function, an independent random variable, and a summary statistic representing the data. In this paper, we introduce a sequential version of Hannig's construction, to be explained shortly, is based on the idea of sampling a complete data set, and which can be made to coincide with the Doob definition of a true parameter value
conditioning on data; see \cite{Doob1949}. It is straightforward to make this strategy `prior free' using the ideas presented in \cite{Fong2023}, leading to a general, sequentially coherent, definition of a fiducial distribution.

In particular, we are able to consider the general framework used by \cite{liang2025}. This assumes observations $Y_{1:n}$, a vector of size $n$, is connected to covariates $X_{1:n}$, a $n\times p$ matrix, and parameter $\theta$, a $p$-dimensional vector, and independent random variable $Z_{1:n}$, via
\begin{equation}\label{liang}
Y_{1:n}=H(X_{1:n},Z_{1:n},\theta),
\end{equation}
with inverse function $G$, such that $\theta=G(Y_{1:n},X_{1:n},Z_{1:n})$. See Assumption 1 and equation (1) in \cite{liang2025}. Our development is to recast (\ref{liang}) into a sequential one-step ahead form and to look directly at a sequence of statistics $(T_m)_{m>n}$ for which we assign a distribution conditional on the observed statistic $T_n$.  
So (\ref{liang}) would become
$T_m=H(X_m,Z_m,T_{m-1}),\quad m>n.$
In short, this is replacing the $\theta$ in (\ref{liang}) with estimator $T_{m-1}$. The limit of $T_m$ as $m\to\infty$ represents a posterior sample from our fiducial distribution, provided it exists, and it this existence which forms the theoretical part of the present paper. This approach mirrors the \cite{Doob1949} definition of a posterior distribution using predictive distributions. The advantage here is that we do not require the need for an inverse function, which to the present day, has been one of the most important and yet severe restrictions on extending the applicability of fiducial inference.

\subsection{The set up}
The fundamental problem of inverse probability for statistical inference can be stated as follows: assume that observed data $x_{1:n}$ are obtained independent and identically distributed according to a probability model $f(x \mid  \theta^*)$ indexed by parameter $\theta \in \Theta$, where the true value $\theta^*$ is unknown. The task is to `invert' the probability distribution on $X$ to obtain a conditional distribution on $\theta$, namely $p(\theta \mid x_{1:n})$, which characterises uncertainty in the unknown value of $\theta^*$ arising from the incomplete information due to the finite sample size, $n$. 

The simplest and most widespread solution to inverse probability is to expand the model to a joint probability on $x$ and $\theta$, treating $\theta$ as if it were a random variable. Following which, straightforward application of conditional probability through Bayes Theorem leads to $p(\theta \mid x) \propto f(x \mid \theta) \, p(\theta)$, where $p(\theta)$ is the marginal density for the parameter. This solves the inverse probability problem but at the expense of introducing a prior. Over a 30-year period, Fisher explored fiducial inference as a solution to inverse probability derived solely from the likelihood function, and some additional assumptions such as sufficiency. 

If $t_n$ is the observed value of a sufficient statistic \(T_n\), being some function of the $x_{1:n}$,  and \(F_n(\cdot \mid \theta)\) is the cumulative distribution function of \(T_n\), then regardless of the philosophical motivation espoused by Fisher, the upshot is the recommended fiducial posterior distribution is given by 
\begin{equation}\label{fisher}
	\Pi(\theta\mid t_n)=1-F_n(t_n\mid \theta),
\end{equation}
whenever it exists as a cumulative distribution function. The constraints are obvious, $F_n(t\mid\theta)$ must be monotone in $\theta$ and be $0$ and $1$ for the relevant extreme values of $\theta$. The corresponding fiducial density function is given by
$\pi(\theta\mid t_n)=-\partial F_n/\partial\theta\,(t_n\mid\theta).$
This seems, and is, quite restrictive. 

Noting that it is possible to derive a $\theta$ sample from an inversion of $F(t\mid\theta)$, so write
$F_n(t\mid\theta)=u \iff H(t,u)=\theta,$
it follows that a sample $\theta$ from the fiducial posterior can be achieved as $H(t_n,U)$ where $U$ is a uniform random variable independent of $T_n$.

The Hannig extension of the Fisher set up is to define a random $\theta$, though we now write this as a random $T$ for generality (which could be a random parameter), as
\begin{equation}\label{hannig}
	T=H(t_n,Z)
\end{equation}
where $Z$ is a r.v. independent of $T_n$. This cuts out the need for the restrictive definition involving distribution functions and can be easily extended to multiple dimensions, though at the cost of an interpretation for (\ref{hannig}).

The aim of the paper is to provide an interpretation and hence a general version of (\ref{hannig}) based on Doob's work on the connection between Bayesian posterior distributions and predictive sampling. The connection between these has previously been exploited in Fong et al. (2023), founded on martingales. We, on the other hand, do not need such a martingale set up.

Our extension of Hannig's notion of a fiducial distribution is based on Doob, and is defined sequentially by
\begin{equation}\label{eq:recursion}
	T_{m+1}=H_m(T_m,Z_m),\quad m\geq n,
\end{equation}
where the $Z_m$ and $T_m$ are independent r.v.
Assuming $T_m$ has a limit $T_\infty$, to exist a.s., it is the $T_\infty$ which is a sample from our fiducial  distribution; i.e.
$T_\infty=\lim_{m\to\infty}T_m$. We can and do include covariate information later on, though for now we define the fiducial distribution without them:

\begin{definition}
	The Doob Fiducial distribution is the distribution of the limit $T_\infty$ which is constructed from (\ref{eq:recursion}) and started at the observed $m=n$ with $T_n=t_n$. Here the $(Z_m)$ are an independent sequence of random variables and the $H_m$ a known sequence of functions.
\end{definition}

\subsection{Demonstrative illustration}
Here we show in a simple illustration how the Fisher fiducial distribution differs from our own. 
As an illustration, if $a>0$ is known and
$f(x\mid \theta)=x^{a-1}\theta^a\,e^{-\theta x}/\Gamma(a),$
then 
$X=G_\theta(Z)$ where $G_\theta(z)=z/\theta$ and $Z$ is a $\mbox{Ga}(a,1)$ random variable. The fiducial density for $\theta$ based on a sample of size $n$ with statistic $t_n=\sum_{i=1}^n x_i/n$ is given by
$\pi(\theta\mid t_n)=\mbox{Ga}(\theta\mid na,nt_n)$.

On the other hand, our approach 
constructs $$f(x_{n+1:N}\mid x_{1:n})=\prod_{m=n+1}^N f(x_m\mid x_{1:m-1})$$
for an arbitrarily large $N$,
where the conditional density functions are based on sampling and updating using the current MLE; i.e.
$$X_{m+1}=Z_m\,T_m/a\quad\mbox{and}\quad T_m=\sum_{i=1}^m X_i/m.$$
We produce the sequence $(T_{m})_{m>n}$, where $T_m=m^{-1}\sum_{i=1}^m X_i$ with
$X_{m+1}=Z_m\,T_m/a$, and the $(Z_m)$ are independent $\mbox{Ga}(a,1)$  variables.
Hence, $T_{m+1}=T_m(m+Z_m/a)/(m+1)$, so completing the term from Definition 1, we have
$H_m(t,z)=t(m+z/a)/(m+1)$. The fiducial distribution is the distribution of the limit $T_\infty$. The rationale and theory behind this strategy will be gone through in the remainder of the paper.

Let us now describe the layout of the paper. 
In section 2 we present some further background to the fiducial distribution of Fisher and provide an interpretation of it which relies on the bootstrap of Efron. This demonstrates a lack of rigor with the idea and indeed highlights Fisher's fiducial distribution as a heuristic approach. We also show how Doob's interpretation of Bayes can be used to define a well motivated fiducial distribution, the key to which is a sequential development of connections between samples and parameters. Section 3 provides the theoretical results which establish conditions under which the Doob fiducial distribution exists. Section 4 contains some  illustrations to make clear the novelty and outline of the new fiducial distribution. These examples include a normal, exponential model and a Weibull and uniform model. 
Section 5 deals with nonparametric regression. An illustration related to logistic regression is provided. 
Finally, section 6 concludes with a brief discussion.


\section{Background}

Before we provide further explanation on the Doob fiducial approach, we show how it is possible to interpret the Fisher fiducial approach as a parametric bootstrap in the spirit of Efron; see \cite{Efron2012}. 
Following the original ideas of Fisher, the construction of the fiducial distribution can be shown to be based on the notion that if 
$T_n(\theta)$ is a statistic produced by taking $n$ i.i.d. samples from $f(\cdot\mid\theta)$, and the observed $t_n$ is arising in the same way with $\theta^*$, then the inequality 
$T_n(\theta)\leq t_n(\theta^*)=t_n$
is suggestive of $\theta< \theta^*$.
Of course this is not true due to stochastic variation. But as a guide, Fisher uses this to derive
$\P(T> \theta\mid t_n)=\P(T_n(\theta)\leq t_n)=F_n(t_n\mid\theta),$
being $T$ the random counterpart of $\theta^*$. 
This can now be converted into the fiducial distribution as in (\ref{fisher}).

To demonstrate this bootstrap interpretation, we consider a Bernoulli model. 
The fiducial distribution is
$$\Pi(\theta\mid t_n)=1-\sum_{k=0}^{t_n} \binom{n}{k} \theta^k(1-\theta)^{n-k}.$$
Such a fiducial distribution is well defined provided that \(t_n<n\) and in such case is Beta$(1+t_n, n-t_n)$. One can see this recalling the binomial expansion of the regularized incomplete beta function, given for instance by \cite{Abram64} in their formula 6.6.4 on page 263, namely:
\[
\frac{\int_0^p t^{m-1}(1-t)^{n-m}\, \d t}{
	\int_0^1 t^{m-1}(1-t)^{n-m}\, \d t}
= \sum_{j=m}^n \binom{n}{j}p^j(1-p)^{n-j},
\]
which holds for \(0<p<1\), \(m=1,\dotsc,n\), and \(n=1,2,\dotsc\).

We can achieve this by taking repeated bootstrap samples $x_{1:n}(\theta)$ as i.i.d. Ber$(\theta)$ for each $\theta$. 
Set $T_n(\theta)$ to be the $\sum_{i=1}^n x_i(\theta)$. Then $T_n(\theta)> t_n$ is evidence of $\theta>\theta^*$. 
So repeating this $B$ times, an Efron style bootstrap estimator for the distribution of $T$ is
$$\Pi^{(B)}(\theta\mid t_n)=B^{-1}\sum_{b=1}^B 1\bigg(T_n^{(b)}(\theta)>t_n\bigg),$$
which for large $B$ converges to $\Pi(\theta\mid t_n)$.

Another similar example can be made with the Poisson distribution. 
Consider the Poisson model and again let 
\(t_n=\sum_{i=1}^n x_i\). In this setting, 
\begin{equation}\label{eq: poisson}
	\Pi(\theta\mid t_n) = 
	\sum_{k=t_n+1}^\infty e^{-n\theta} \frac{(n\theta)^k}{k!}
	= \frac{n^{t_n+1}}{\Gamma(t_n+1)}
	\int_0^\theta e^{-nx}\, x^{t_n}
	\,\d x
\end{equation}
so that the fiducial distribution is Gamma\((1+t_n,n)\). 
The second identity in \eqref{eq: poisson} is a consequence 
of a well known relationship between the Gamma and the Poisson distribution \citep[][formula 4.96, page 197]{Johnson05}, namely \(\P(X \leq x)=\P(Y\geq m)\) if 
$X$ is Gamma$(m,\beta)$, being $m$ an integer, $\beta>0$ and  $Y$ is Poisson$(\beta x)$ with $x>0$.

The motivation behind the Hannig extension is grounded in the notion that it is often possible to write
$T_n=G_{Z}(\theta^*)$ 
where $Z$ is a random variable whose distribution does not depend on the true value $\theta^*$ of the parameter. 
In this case, assuming an inversion is possible, either exactly or numerically, one has
$\theta^*=G^{-1}_Z(T_n).$
The fiducial trick here is to switch what is random and what is fixed. So the idea is that $T_n$ is replaced by its observed value $t_n$ and becomes fixed whereas $\theta^*$ is replaced by its random counterpart $T$. So, the fiducial distribution is the distribution of $T=G_Z^{-1}(t_n)$. 

Hence, we see that both Fisher's original definition and the Hannig extension have some interpretation concerns. On the other hand, \cite{Doob1949} has adequately explained how $\theta^*$ can be recovered from an infinite sample, which is where we start our version of the fiducial distribution. So Doob showed, for identifiable parameters, that it is possible to write $\theta^*=H(X_{1:\infty})$. This is based on the fact that 
$\lim_{n\to\infty} 
\frac{1}{n} \sum_{i=1}^n 1(X_i\leq x) 
=F(x\mid\theta^*)\,\mbox{a.s.}$
So $\theta^*$ is available from the infinite random sample.

Since $x_{1:n}$ are observed and $t_n$ is a sufficient statistic, it follows that for all $m$ it is that 
\begin{equation}\label{doob}
	T=\widetilde{H}_m(T_m,X_{m+1:\infty})
\end{equation}
for some $\widetilde{H}_m$. 
An appropriate choice 
for the conditional distribution of  $X_{m+1:\infty}$ given $T_m=t_m$ is such that for every \(N>m\) the conditional density of $x_{m+1:N}$ given $T_m=t_m$ is 
$\prod_{l=m+1}^N f(x_{l}\mid t_{l-1}).$
A sequential version of (\ref{doob}) is therefore given by
$T_{m+1}=H_m(T_m,Z_m)$ for some sequence of functions $(H_m)$
where the assumption is that $Z_m$ is independent of $T_m$; i.e. the $X_{1:m}$.

\section{General Models}

In this section we consider the convergence of the sequence (\ref{eq:recursion}) to a limit random variable $T_\infty$. 
In the following we consider \(H_m(t, z) = \varphi^{-1}(\varphi(t) + g_m(t,z))\). This is a very general expression for \(H_m\). For example, for the gamma model \(\varphi(t)=\log(t)\) and \(g_m(t,z)=\log((m+z/a)/(m+1))\).

\begin{theorem}\label{th: gen}
	Assume that there exist a function \(\varphi\) of a single variable  and a function \(g\) of two variables such that:
	\begin{enumerate}[i)]
		\item \(\varphi\) is an invertible function with a closed range and \(\varphi^{-1}\) is continuous
		\item the following identity holds for every \(t\) and \(z\): 
		\begin{equation}\label{eq: cond_gen}
			H_m(t, z) = \varphi^{-1}(\varphi(t) + g_m(t,z))
		\end{equation}
		\item the following two random series converge almost surely:
		\begin{equation}\label{eq: series_cond}
			\begin{split}
				\sum_{m=1}^\infty 
				&\E(g_m(T_m,Z_m)\mid T_m), \\ 
				\sum_{m=1}^\infty 
				&\E(g_m(T_m,Z_m)^2\mid T_m),
			\end{split}
		\end{equation}
		\item \((T_m)_{m=1}^\infty\) and \((Z_m)_{m=1}^\infty\) are two random sequences such that the sequence 
		\((T_1,Z_1, Z_2, \dotsc)\) is an independent sequence, 
		and the following recursive identity holds
		\begin{equation}\label{eq: recursion2}
			T_{m+1}=H_m(T_m,Z_m),\quad m\geq 1,
		\end{equation}
	\end{enumerate}
	Then there exists a random variable \(T\) such that \(T_m\to T\), as \(m\to \infty\), almost surely.
\end{theorem}

A general multivariate parameter form for this, with one dimensional data, has, for $j=1,\ldots,d$,
$T_{m+1,j}=H_{m,j}(T_m,Z_m),$
where $T_m=(T_{m,1},\ldots,T_{m,d})$ and $Z_m$ is a one dimensional innovation variable for which 
$T_{m+1,j}=\varphi_j^{-1}(\varphi_j(T_{m,j})+g_m(T_m,Z_m)).$

\subsection{Proof of Theorem \ref{th: gen}}

We resort to the following result whose proof is given for instance by \cite{Durrett19} on page 208, Theorem 4.5.2.  
\begin{theorem}\label{th: durrett}
	Assume that \((X_n)_{n=0}^\infty\) is a martingale with respect a filtration \((\mF_n)_{n=0}^\infty\) such that \(X_0=0\) and \(\E(X_n^2)<\infty\) for all \(n\). Moreover, let \((A_n)_{n=1}^\infty\) be the increasing process associated with \((X_n)_{n=0}^\infty\), namely:
	\begin{equation}\label{eq: A_n}
		A_n= \sum_{m=1}^n \E(X_m^2\mid \mF_{m-1})-X_{m-1}^2
		=  \sum_{m=1}^n \E((X_m-X_{m-1})^2\mid \mF_{m-1}),
	\end{equation}
	and let \(A_\infty = \lim_{n\to \infty} A_n\).
	
	Then \(\lim_{n\to \infty} X_n\) exists and is finite almost surely on \(\{A_\infty <\infty\}\).
\end{theorem}

An immediate consequence of Theorem \ref{th: durrett} is the following Corollary:
\begin{corollary}\label{cor}
	Let \((\mF_n)_{n=0}^\infty\) be a filtration and let \((Y_n)_{n=1}^\infty\) be a sequence of square integrable random variables such that \(Y_n\) is measurable with respect to \(\mF_n\), for every \(n\geq 1\). If the following two random series converge almost surely:
	\begin{equation}\label{eq: cor_cond}
			\sum_{m=1}^\infty 
			\E(Y_{m}\mid \mF_{m-1})\quad\mbox{and}\quad
			\sum_{m=1}^\infty 
			\E(Y_{m}^2\mid \mF_{m-1}),
	\end{equation}
	then \(\sum_{n=1}^\infty Y_n\) converges almost surely. 
\end{corollary}
\begin{proof}[Proof of Corollary \ref{cor}]
	Let \(X_0=0\) and 
	\begin{equation}\label{eq: Xn}
		X_n= \sum_{m=1}^n \{Y_m-\E(Y_m\mid \mF_{m-1})\}
	\end{equation}
	for \(n=1,2,\dotsc\)
	The random sequence \((X_n)_{n=1}^\infty\) is a square integrable martingale such that 
	$
	X_{n}-X_{n-1}= Y_n- \E(Y_n\mid \mF_{n-1}),
	$
	for \(n=1,2,\dotsc\), 
	and therefore the increasing process \((A_n)_{n=1}^\infty\) associated with \((X_n)_{n=1}^\infty\) is given by 
	$
	A_n=\sum_{m=1}^n \var(Y_m\mid \mF_{m-1}).
	$
	Being the second series in \eqref{eq: cor_cond} almost surely convergent, we obtain:
	\begin{equation*}
		A_\infty=\sum_{m=1}^\infty \var(Y_m\mid \mF_{m-1})
		\leq \sum_{m=1}^\infty 
		\E(Y_{m}^2\mid \mF_{m-1}) <\infty.
	\end{equation*}
	Therefore, we can apply Theorem \ref{th: durrett} to obtain that \(X_n\) converges almost surely. 
	By \eqref{eq: Xn}, this means that the following random series 
	\begin{equation*}
		X_\infty = 
		\sum_{m=1}^\infty \{Y_m-\E(Y_m\mid \mF_{m-1})\}    
	\end{equation*}
	converges almost surely. Combining this fact with the almost sure convergence of the first random series in \eqref{eq: cor_cond} we obtain that 
	the random series 
	\[\sum_{m=1}^\infty Y_m=
	X_\infty + 
	\sum_{m=1}^\infty \E(Y_m\mid \mF_{m-1})\]
	converges almost surely. 
\end{proof}

The following is a well known result, whose proof is given for instance by \cite{Durrett19} on page 182, Example 4.1.7.

\begin{lemma}\label{lemma}
	Suppose that \(X\) and \(Y\) are two independent random variables or random vectors. Let \(\xi\) be a function such that \(\E(|\xi(X,Y)|)<\infty\) and let \(\phi(x)=\E(\xi(x,Y))\). Then
	$\E(\xi(X,Y)\mid X) = \phi(X).$
\end{lemma}

We are now ready to prove Theorem \ref{th: gen}. 
Combining \eqref{eq:recursion} and \eqref{eq: cond_gen} we have that 
$
\varphi(T_{m+1})-\varphi(T_m) = g_m(T_m, z_m)   
$
for every \(m\geq n\). Therefore, 
\begin{equation}\label{eq: phiTm}
		\varphi(T_m) = \varphi(T_n) + 
		\sum_{\ell=n}^{m-1}  
		\{\varphi(T_{\ell+1})-\varphi(T_\ell)\}
		=\varphi(T_n) + 
		\sum_{\ell=n}^{m-1} g_\ell(T_\ell, z_\ell), 
\end{equation}
for every \(m\geq n\).

For every \(m\geq 1\), let \(\mF_m\) be the sigma-field generated by the  \(T_1, Z_1, \dotsc, Z_m\). 
Applying inductively the recursion \eqref{eq: recursion2}, one can see that \(T_m\) is a function of \(T_1, Z_1, \dotsc, Z_{m-1}\), say 
\(T_m=h_m(T_1, Z_1, \dotsc, Z_{m-1})\) for some function \(h_m\), 
and therefore \(T_m\) and \(Z_m\) are independent. 
Hence, setting 
$
\phi_m(t)=\E(g_m(t,Z_m)),
$
for every \(t\), by Lemma \ref{lemma} we have that 
\begin{equation}\label{eq: proof-cor1}
	\begin{split}
		\E(g_m(T_m,Z_m)&\mid \mF_{m-1})\\
		&= \E(g_m(h_m(T_1, Z_1, \dotsc, Z_{m-1}),Z_m)
		\mid T_1, Z_1, \dotsc, Z_{m-1})\\
		&= \phi_m(T_m),
	\end{split}
\end{equation}
and applying again Lemma \ref{lemma} we have that
\begin{equation}\label{eq: proof-cor2}
	\E(g_m(T_m,Z_m)\mid T_m)=\phi_m(T_m).
\end{equation}
Combining \eqref{eq: proof-cor1} and 
\eqref{eq: proof-cor2} we have that
\[
\E(g_m(T_m,Z_m)\mid \mF_{m-1})
= \E(g_m(T_m,Z_m)\mid T_m),
\]
and similarly we can show that 
\[\E(g_m(T_m,Z_m)^2\mid \mF_{m-1})
= \E(g_m(T_m,Z_m)^2\mid T_m).\] 
So the two random series \eqref{eq: series_cond} are both equal to the
\(\sum_{m=1}^\infty 
\E(g_m(T_m,Z_m)\mid \mF_m),\) and  
\(\sum_{m=1}^\infty \E(g_m(T_m,Z_m)^2\mid \mF_m)\), respectively. 
Since the two random series converge almost surely, we can apply Corollary \ref{cor} setting 
\(Y_m=g_m(T_m,Z_m)\). This proves that 
the random series 
$
\sum_{m=n}^{\infty} g_m(T_m, Z_m)
$
converges almost surely. 
Therefore, by \eqref{eq: phiTm}, \(\varphi(T_m)\) 
converges almost surely to some random variable 
\(W\). Being \(\varphi^{-1}\) continuous by assumption, we can apply the continuous mapping theorem to obtain that \(T_m\) converges almost surely to \(\varphi^{-1}(W)\) and the proof of Theorem \ref{th: gen} is complete.

\section{Illustrations}

Here we present illustrations with the aim of demonstrating the Doob fiducial distributions based on the representation of posterior samples. We have modified Doob's predictive distributions to leave out a prior construction while maintaining the key property of a martingale, which is all Doob required to establish the existence of a posterior from predictive sampling.  

\subsection{Normal model}\label{ss:normal} A normal model with unknown mean $\theta$ and known variance $\sigma$, and observed $y_{1:n}$, so let $t_n=\bar{y}$. Then, according to the sampling scheme for $Y_{n+1:\infty}$
where $Y_{m+1}\mid T_m \sim N(T_m,\sigma^2)$, it is that
$T_{m+1}=T_m+Z_m\,\sigma/(m+1)\,,\quad m\geq n,$
where the $(Z_m)$ are independent standard normal. 
So
$$T=t_n+\sigma\,Z\,\sqrt{\sum_{m=n+1}^\infty 1/m^2},$$
where $Z$ is a standard normal random variable.
Hence, the distribution of 
$h(t_n,Z)$ is our fiducial distribution where:
$$h(t,z)=t+\sigma\,z\,\sqrt{\sum_{m=n+1}^\infty 1/m^2}.$$

It is possible to bound the series appearing in this fiducial distribution using telescopic series. Indeed, 
\begin{align*}
	\sum_{m=n+1}^\infty \frac{1}{m^2} &<  
	\sum_{m=n+1}^\infty \frac{1}{m(m-1)} =
	\sum_{m=n+1}^\infty \left(\frac{1}{m-1}-\frac{1}{m}\right) = 
	\frac{1}{n}\\
	\sum_{m=n+1}^\infty \frac{1}{m^2} &>  
	\sum_{m=n+1}^\infty \frac{1}{m(m+1)} =
	\sum_{m=n+1}^\infty \left(\frac{1}{m}-\frac{1}{m+1}\right) = 
	\frac{1}{n+1}\\
\end{align*}
This implies that our fiducial distribution is close to the one of Fisher which is \(N(t_n, \sigma^2/n)\), especially if \(n\) is large.

For a multivariate case, consider the case when $\theta$ and $\sigma$ are unknown. The start is $T_n$ to be the sample mean and $S_n^2$ to be the sample variance.
Then for a general $m$, 
$Y_{m+1}=T_m+S_m\,Z_m$
where $(Z_m)$ is a sequence of independent standard normal random variables.  
The updates are
$T_{m+1}=T_m+S_m\,Z_m/(m+1)$
and
$S_{m+1}^2=S_m^2+S_m^2\left\{Z_m^2/(m+1)-1/m\right\}.$
So
$$g_{m,1}(t,z)=z\,\sqrt{t_2}/(m+1)\quad\mbox{and}\quad g_{m,2}(t,z)=t_2\,(z^2/(m+1)-1/m),$$
where $t=(t_1,t_2)$, the components being the sample mean and the sample variance.
Now, given two random variables $T=(T_1,T_2)$ and $Z$, we have that
$$\E(g_{m,1}(T,Z)\mid T=t)=0\quad\mbox{and}\quad 
\E(g_{m,2}(T,Z)\mid T=t)=-\frac{t_2}{m(m+1)}.$$
Also
$$E(g^2_{m,1}(T,Z)\mid T=t)=\frac{t_2}{(m+1)^2}\quad\mbox{and}\quad 
E(g^2_{m,2}(T,Z)\mid T=t)\leq c\frac{t_2^2}{m^2},$$
for some constant $c>0$. So, as long as $t_2$ is bounded for all $m$, the sums are finite. We can bound these by taking
$t_{m,2}$ to be  the max of some large constant and the generated $t_{m,2}$.

For an illustration, $n=50$ and take $T_n=0$ and $S_n^2=1$.
The fiducial posterior distributions are presented  in Fig.~\ref{fignorm}. In Fig.~\ref{figcorr} the plots of the samples are given, strongly suggesting an independence property; indeed, the sample correlation is 0.04.
This is important since the sample mean and sample variance are independent.

\begin{center}
	\begin{figure}[!htbp]
		\begin{center}
			\includegraphics[width=14cm,height=5cm]{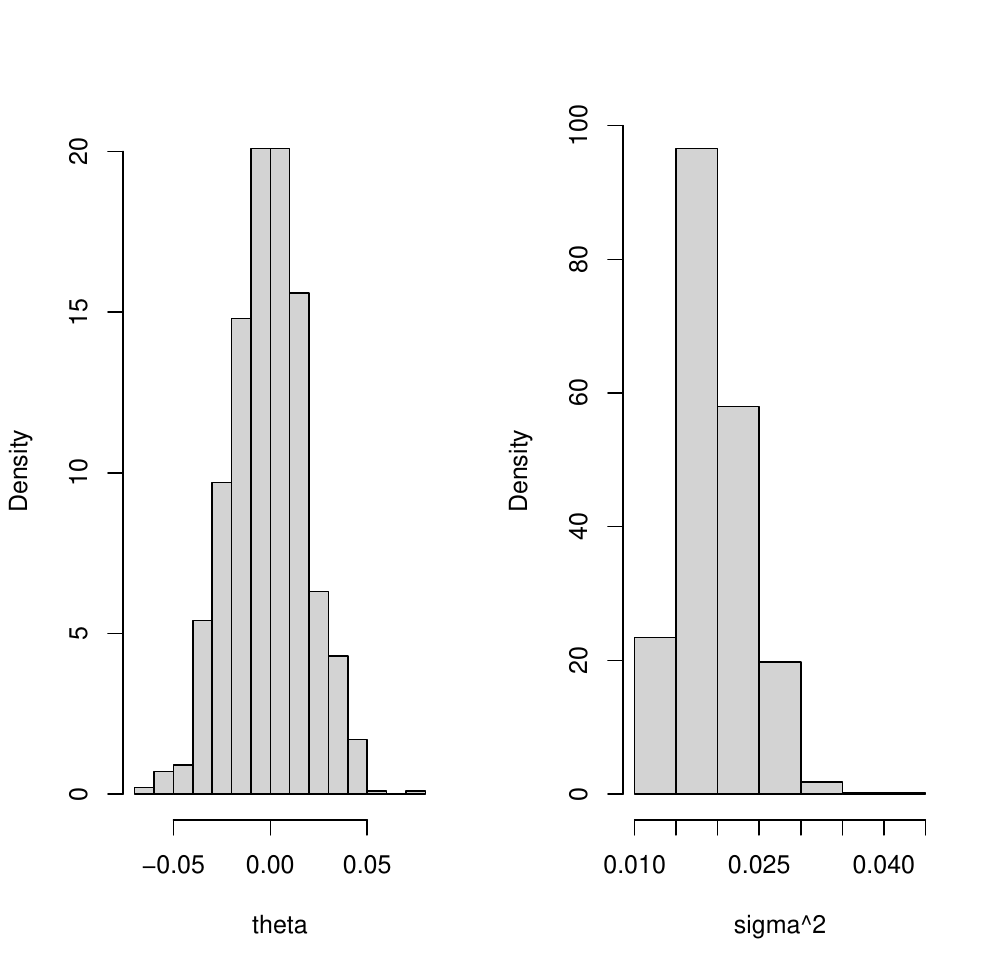}
			\caption{Fiducial posterior distributions for mean and variance parameters from Section \ref{ss:normal}}
			\label{fignorm}
		\end{center}
	\end{figure}
\end{center} 

\begin{center}
	\begin{figure}[!htbp]
		\begin{center}
			\includegraphics[width=12cm,height=5cm]{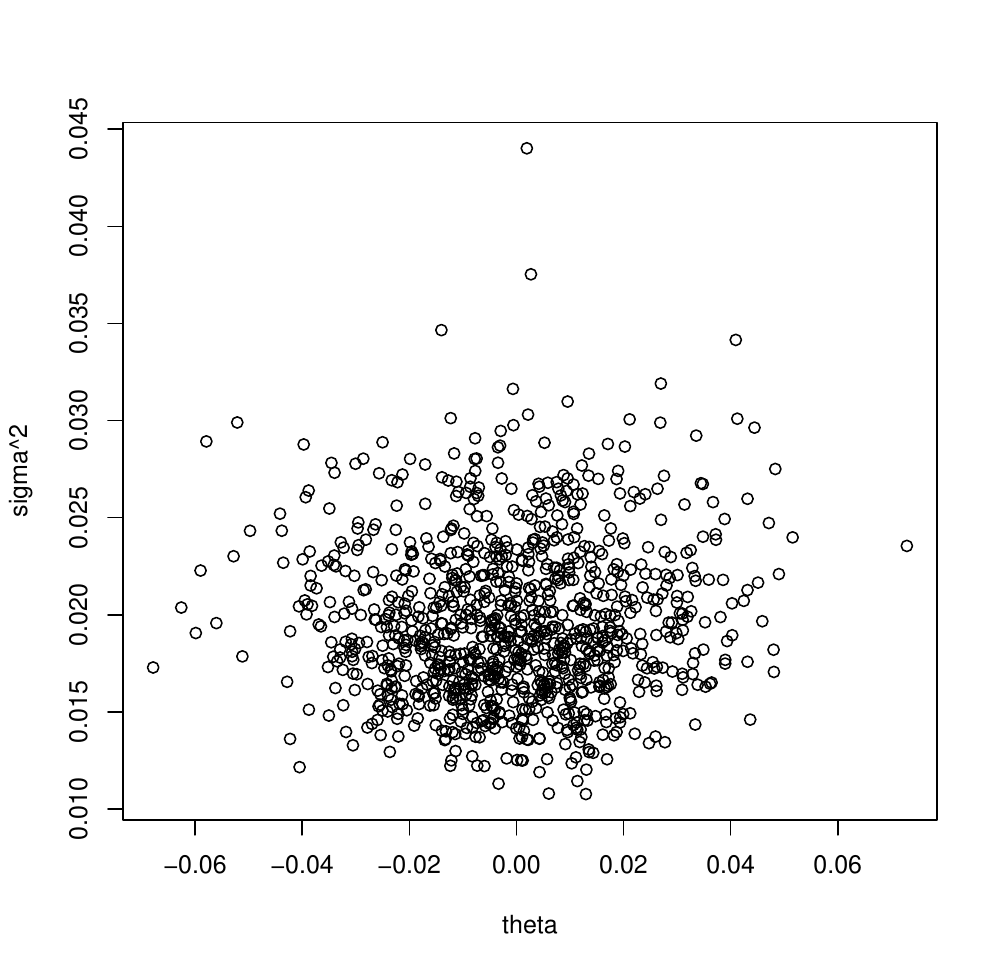}
			\caption{Fiducial posterior samples for mean and variance parameters from Section \ref{ss:normal}}
			\label{figcorr}
		\end{center}
	\end{figure}
\end{center} 

\subsection{An exponential example}\label{ss:exponential}

Assume $x_{1:n}$ are i.i.d. realizations from the exponential density $f(x\mid\theta)=\theta^{-1}\exp(-x/\theta)$. The sufficient statistic is $t_n=\bar{x}$. 

The Fisher fiducial density function is given by
an inverse gamma density with parameters $(n,n\theta_n)$.
This is also the posterior arising from the objective prior proportional to $1/\theta$. For the details, the statistic $t_n=\sum_{i=1}^n x_i$ is a realization of a gamma random variable with parameters $(n,1/\theta)$ so
the fiducial density is 
$\pi(\theta\mid x_{1:n})=-\frac{\partial}{\partial\theta}
\int_0^{t_n/\theta} s^{n-1}e^{-s}\,ds/\Gamma(n)$
which becomes the said inverse gamma density.

Our fiducial posterior is based on sampling $x_{n+1:\infty}$ via $x_{m+1}$ is taken from $f(x\mid t_m)$ and
$t_m=\bar{x}_m$. Hence, a sample from the posterior is given by
$t_\infty=t_n\,\prod_{i=n}^\infty (i+z_i)/(i+1),$
where the $(z_i)$ are realizations of independent standard exponential random variables. So $t_\infty=t_n\,z$
where 
$z=\prod_{i=n}^\infty (i+z_i)/(i+1),$
which is a realization of a finite r.v.

\begin{center}
	\begin{figure}[!htbp]
		\begin{center}
			\includegraphics[width=12cm,height=5cm]{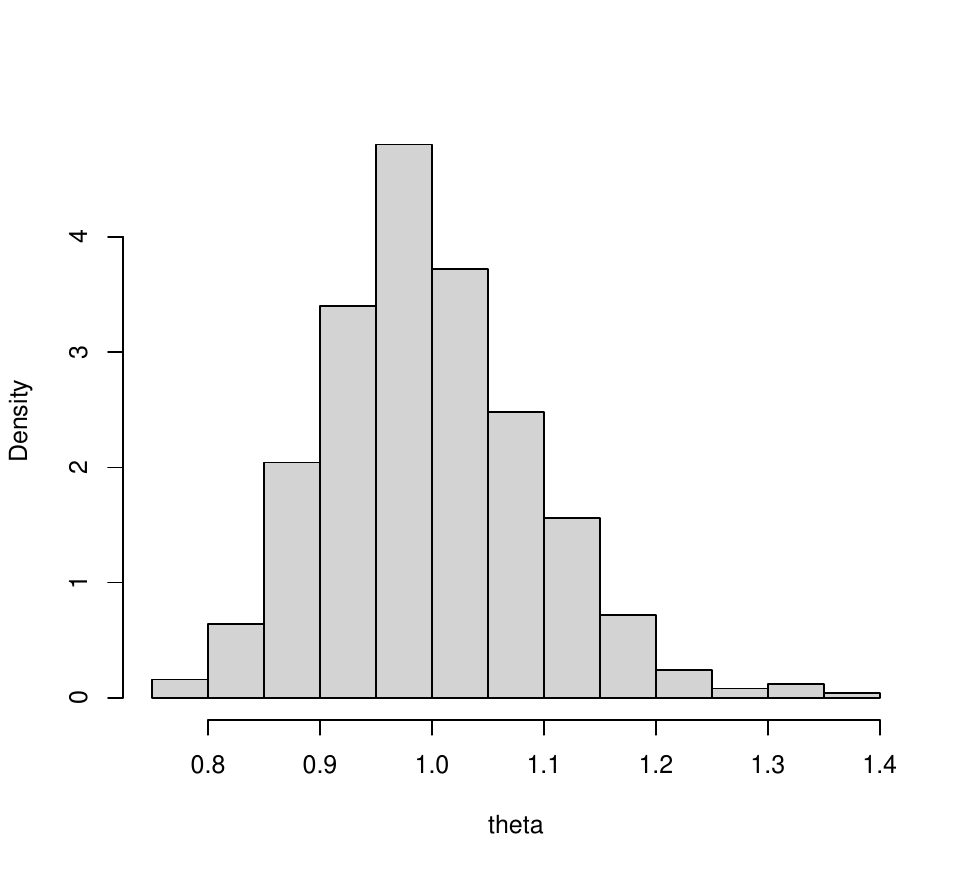}
			\caption{Fiducial posterior for exponential model from Section \ref{ss:exponential}}
			\label{figb}
		\end{center}
	\end{figure}
\end{center} 

See Fig.~\ref{figb}, where $n=100$ and $\theta_n=1$, for a fiducial posterior distribution for the exponential model.

To demonstrate almost sure convergence we can use Kakutani's theorem on product martingales, which we copy here.

\begin{theorem}\label{th: kakutani}
	Let $X_1, X_2,\ldots$ be independent non-negative r.v. each of mean 1 and define $M_0=1$ and, for $n\geq 1$, $M_n=X_1X_2\cdots X_n$. Then $M_n$ is a non-negative martingale, so that $M_\infty=\lim M_n$ exists a.s.
	The following statements are equivalent:
	\begin{enumerate}
		\item $M_n\to M_\infty$ in $L_1$.
		\item $\E(M_\infty)=1$.
		\item $\prod_{n} a_n>0$ where $0<a_n=\E(X_n^{1/2})\leq 1.$
		\item\label{eq: kakutani-cond} $\sum_n(1-a_n)<\infty$.
	\end{enumerate}
\end{theorem}

To show \ref{eq: kakutani-cond}.\  
holds for Kakutani's theorem on product martingales, write
$$1-a_m=\sqrt{\frac{m}{m+1}}\int_0^\infty e^{-z}\left\{\sqrt{1+1/m}-\sqrt{1+z/m}\right\}\,\rm{d} z.$$
It can be shown by standard arguments that $\sum_m(1-a_m)<\infty$. 
The term of order $1/m$ disappears since 
$\frac{1}{2}\int_0^\infty (1-z)e^{-z} \rm{d}z=0$. 
Hence, 
the product martingale defined by
$M_n=\prod_{i=1}^n (i+Z_i)/(i+1),$
where $Z_1,Z_2,\dotsc$ are independent standard exponential random variables,  
converges with probability one to $M_\infty$ and
$\E(M_\infty)=1$ and $M_n\to M_\infty\,\,\mbox{in }L_1.$
In the next section we will look at more general theory for convergence. 


\subsection{Weibull model}\label{ss:weibull}
Consider a Weibull model of the form
$f(x\mid\theta)=\theta x^{\theta-1}\,\exp\left(-x^\theta\right).$
There is no statistic here which summarizes the data based on a sample of size $n$.
However, for a single sample $x$ from the Weibull model, it is possible to express it as
$x=z^{1/\theta}$
where $z$ is a realization of a standard exponential random variable.
A problem of inverting $x=z^{1/\theta}$ is that if $x<1$ then $z<1$ is required for a $\theta>0$ to exist, and if $x>1$ then $z>1$ is required. See \cite{Hannig2016}. Hence, a sample from the \cite{Hannig2016} fiducial distribution is
$\theta=\log (z)/\log (x),$
where $z$ is a realization from a standard exponential constrained to be greater than 1 if $x>1$ while it is constrained to be less than 1 if $x<1$. This yields
$\pi_H(\theta\mid x)\propto x^\theta\,\exp(-x^\theta)$
which is the likelihood multiplied by $1/\theta$, and this $1/\theta$ can be identified as the Jeffreys prior. 

With a sample of size $n$ one will have 
$x_i=z_i^{1/\theta}$, $i=1,\ldots,n$
where the same rules apply to each $(x_i,z_i)$ as before. The problem now is that due to a lack of a statistic corresponding to the parameter, a least squares solution is provided as part of the inversion procedure; i.e.
$\theta$ minimizes the least squares
$\sum_{i=1}^n \left(\log x_i-\theta^{-1}\log u_i\right)^2$
so
$$\theta=\frac{\sum_{i=1}^n (\log u_i)^2}
{\sum_{i=1}^n (\log x_i)\,(\log u_i)}.
$$
The distribution of $\theta$ arises from the distributions of the $(u_i)$ and the solution can be found, appearing as Theorem 1 in \cite{Hannig2016}. The form is as a likelihood times prior, where the prior appears in Equation (4) of \cite{Hannig2016} and can possibly include the data itself.  
This is a somewhat unsatisfactory outcome to a principled idea first developed by Fisher.

On the other hand, here we show that a sequential interpretation of the true parameter value, as motivated by Doob, and to be defined via missing data, yields a principled fiducial distribution. This is because we only ever utilize the generation of missing data via the $x=z^{1/\theta}$ construct and even then only move from a given estimate of $\theta$ and sampled $z$ to a $x$.

The updates based on estimate $t_n$ are, for $m\geq n$,
$$
t_{m+1}=t_m+\frac{1}{m+1}\,s(x_{m+1};t_m)
$$
where 
$s(x;\theta)=1/\theta+\log x-x^\theta\,\log x,$
is the score function, the $(x_m=z_m^{1/t_{m-1}})$ and the $(z_m)$ are realizations of i.i.d. standard exponential. These updates are based on 
the score functions; i.e.
$s(x;\theta)=(\partial/\partial\theta)\log f(x\mid\theta).$ See \cite{holmes2023} for the motivation for these updates.

\begin{center}
	\begin{figure}[!htbp]
		\begin{center}
			\includegraphics[width=12cm,height=5cm]{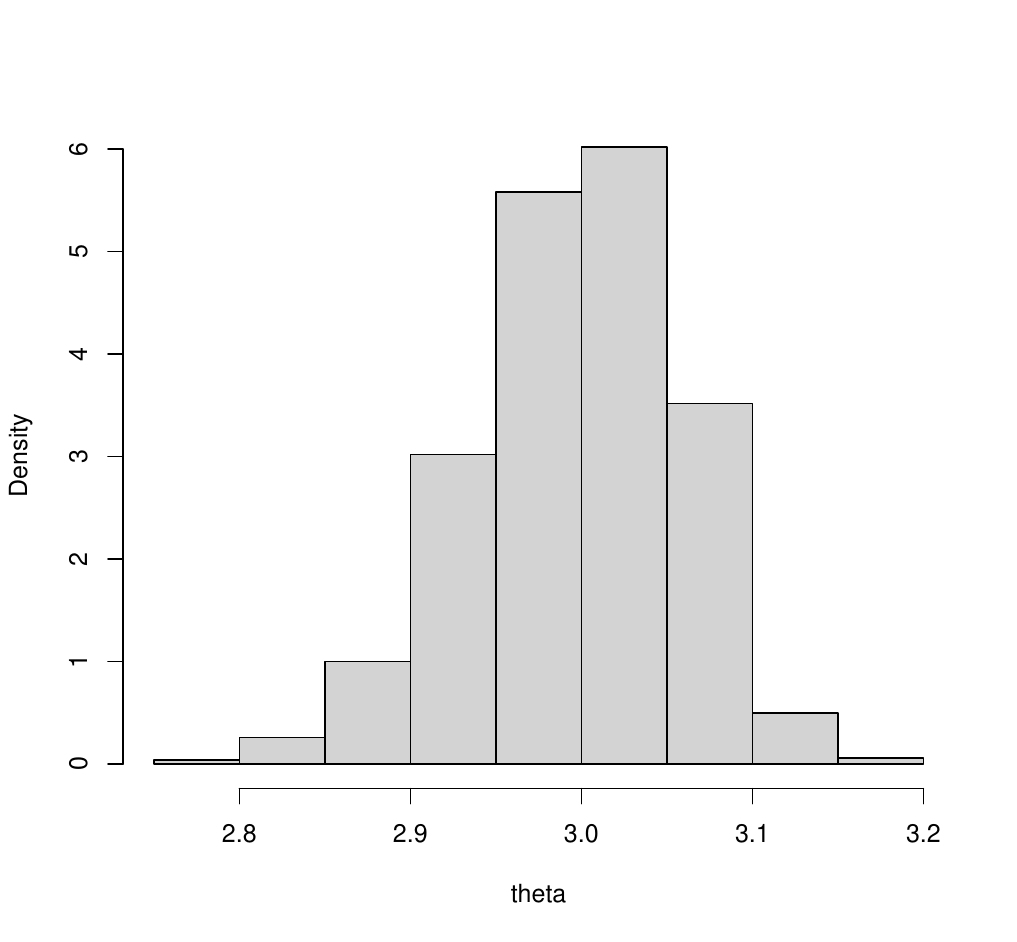}
			\caption{Fiducial posterior for Weibull model from Section \ref{ss:weibull}}
			\label{figw}
		\end{center}
	\end{figure}
\end{center} 

Assuming we have a sample of size $n=50$ with $\theta_n=3$, Fig.~\ref{figw} shows a sample of $\theta$s from our fiducial distribution, by running each sequence up to a size $m=1000$. 

\subsection{Uniform model}\label{ss:uniform}
Consider the uniform distribution on \([0,\theta]\). In this setting, a suitable estimate which is unbiased and a sufficient statistic is 
\(t_n= (n+1)/n\, x_{(n)}\), 
where \(x_{(n)}=\max\{x_1,\dotsc,x_n\}\). 
According to the sampling scheme where \(x_{m+1}=t_m\,u_m\) where \(u_m\) is a realization of a uniform random variable $U_m$ on \([0,1]\) independent of \(T_m\), we have that
\[
t_{m+1}= 
\frac{m+2}{m+1} \left( \frac{m}{m+1}\vee u_m\right) t_m
\]
where \(a \vee b\) denotes \(\max\{a,b\}\).

\begin{center}
	\begin{figure}[!htbp]
		\begin{center}
			\includegraphics[width=12cm,height=5cm]{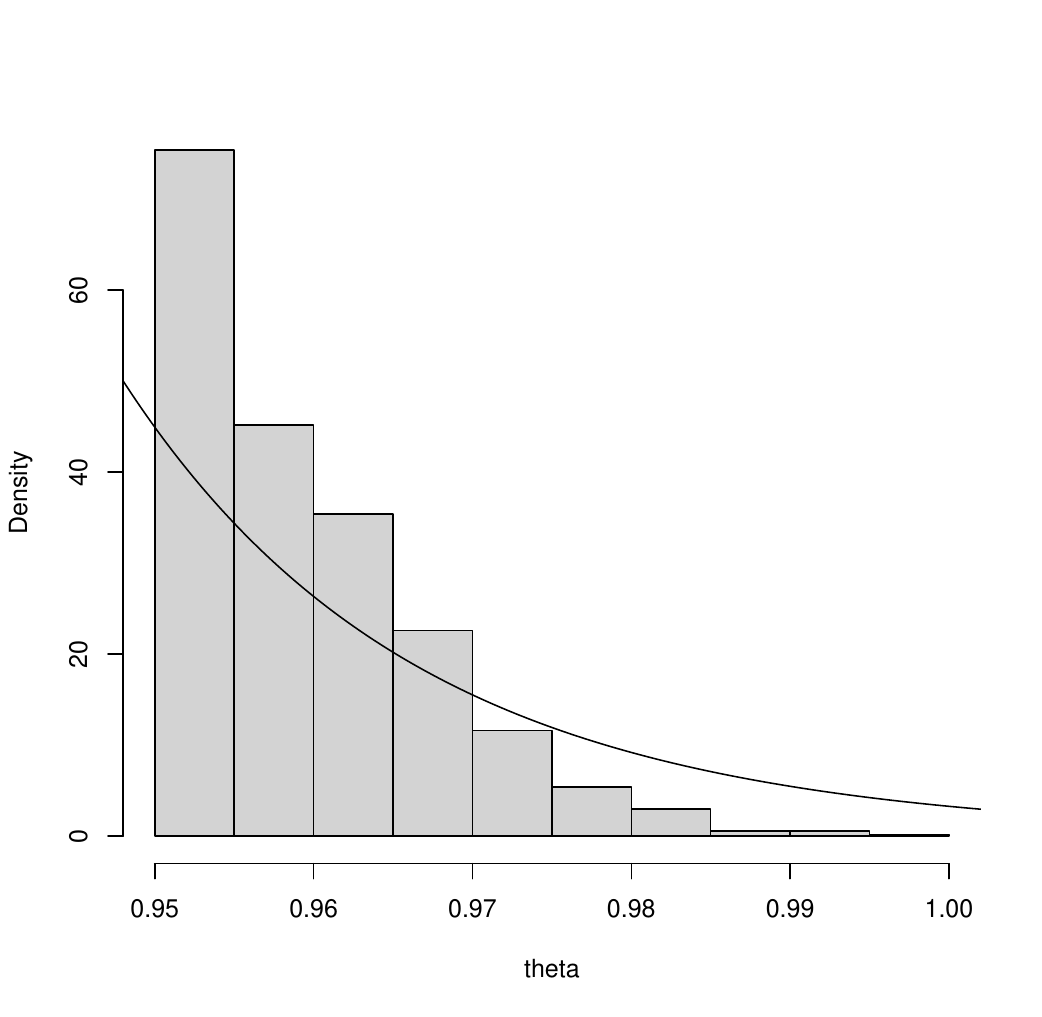}
			\caption{Fiducial posterior for Uniform model (histogram) from Section \ref{ss:uniform} alongside the Fisher fiducial distribution (line)}
			\label{figu}
		\end{center}
	\end{figure}
\end{center} 

Fig.~\ref{figu} shows our fiducial distribution as a histogram, taking $n=50$ and the $x_{(n)}=0.947$, with the true $\theta$ being 1. The Fisher fiducial distribution is shown as a line and is the Pareto density $\pi(\theta\mid x_{(n)})=n x_{(n)}^n/\theta^{n+1}\,1(\theta>x_{(n)})$.  

Almost sure convergence of \(T_m\) to a positive random variable as \(m\) diverges follows from Kakutani's Theorem 
\ref{th: kakutani}. 
To show this, set
\begin{align*}
	\mu_m&=\frac{m+2}{m+1}\E\left( \frac{m}{m+1}\vee U_m\right)
	= \frac{m+2}{m+1}\,\frac{m}{m+1}
	\left(1+\frac{1}{2m(m+1)}\right)\\
	X_m&= \frac{1}{\mu_m}\frac{m+2}{m+1}\left( \frac{m}{m+1}\vee U_m\right)
\end{align*}
so that \(\E(X_m)=1,\), and 
$
T_{m+1} = T_n \prod_{j=n}^m (\mu_j X_j),
$
for every \(m\geq n\). 
Moreover,
\begin{align}
	\E\left(\sqrt{X_m}\right)
	&= \frac{1}{\sqrt{\mu_m}}\,
	\sqrt{\frac{m+2}{m+1}}\, \nonumber
	\E\left( \sqrt{\frac{m}{m+1}}\vee \sqrt{U_m}\right)\\
	&= \frac{\frac{2}{3}\sqrt{\frac{m+1}{m}}+\frac{1}{3}\frac{m}{m+1}}{\sqrt{1+\frac{1}{2}\left(\frac{1}{m}-\frac{1}{m+1}\right)}} \label{eq: step1}
\end{align}
Since the numerator in \eqref{eq: step1} is bounded from below by one by the inequality of arithmetic and geometric means, then the following is bounded from above by the term of a (telescopic) converging sequence:
\[
1-\E\left(\sqrt{X_m}\right) \leq 
1-\frac{1}{1+\frac{1}{2}
	\left(\frac{1}{m}-\frac{1}{m+1}\right)}\leq
\frac{1}{2}
\left(\frac{1}{m}-\frac{1}{m+1}\right)
\]
By Kakutani Theorem \ref{th: kakutani}, the product 
\(\prod_{m=n}^\infty X_m\) converges almost surely to a positive random variable and its expectation is equal to one. Moreover,
\[
\prod_{m=n}^\infty \mu_m = \frac{n}{n+1}
\exp\left\{\sum_{m=n}^\infty 
\log\left(1+\frac{1}{2}
\left(\frac{1}{m}-\frac{1}{m+1}\right)\right)\right\}
\]
So that by the inequality \(\log(1+t)\leq t\), we have that
\[\frac{n}{n+1} \leq \prod_{m=n}^\infty \mu_m 
\leq \frac{n}{n+1}\,e^{1/(2n)}.\]
As a consequence, \(T_m\) converges almost surely to a positive random variable \(T_\infty\) as \(m\) diverges to infinity such that
$
\E(T_\infty\mid T_n=t_n) = t_n\,\prod_{m=n}^\infty \mu_m.
$
This means that the expectation of our fiducial distribution is between \(x_{(n)}\) and \(x_{(n)}\,e^{1/(2n)}\). 
Instead, the mean of Fisher fiducial distribution is \(n/(n-1)x_{(n)}\). Of course, the difference between these two estimators is negligible for large \(n\).

Our approach allows us to find the fiducial distribution of both parameters of the uniform distribution. Consider the uniform distribution on \([a-b,a+b]\). In this setting, unbiased estimates for the parameters \(a\) and \(b\) based on sufficient statistics are 
\begin{align*}
	a_n&= \frac{1}{2} (\min\{x_1,\dotsc,x_n\}+
	\max\{x_1,\dotsc,x_n\})\\
	b_n&= \frac{1}{2}\, \frac{n+1}{n-1}\, 
	\left(\max\{x_1,\dotsc,x_n\}-
	\min\{x_1,\dotsc,x_n\}\right)
\end{align*}
We consider the sampling scheme where 
\(x_{m+1} = a_m-b_m +2b_m\,u_m\), \(u_m\) is a realization from a uniform random variable \(U_m\) on \([0,1]\) and generated independently of \((a_m,b_m)\), and \(U_n,U_{n+1},\dotsc\) are independent. If we denote by 
\(x^+=x\vee 0\) the positive part of \(x\) for every \(x\in\BR\), the updates are:
\begin{align}
	\label{eq: uniform_updatingA}
	a_{m+1}&=a_m +b_m\left[\left(u_m-\frac{m}{m+1}\right)^+
	-\left(\frac{1}{m+1}-u_m\right)^+\right]  \\
	\begin{split}
		b_{m+1}&=b_m\left(1-\frac{2}{m(m+1)}\right)
		\label{eq: uniform_updatingB}
		\\ 
		&\phantom{XX}+b_m \frac{m+2}{m}
		\left[\left(u_m-\frac{m}{m+1}\right)^+
		+\left(\frac{1}{m+1}-u_m\right)^+\right],\end{split}
\end{align}
for \(m\geq n\).

If we let \(w_m=(1/(m+1)-u_m)^+\) and 
\(v_m= (u_m-m/(m+1))^+\), we have that 
\begin{align*}
	a_{m+1}&=a_m+b_m(v_m-w_m)\\
	b_{m+1}&= b_m\left\{1-\frac{2}{m(m+1)}+\frac{m+2}{m}(v_m+w_m)\right\}.
\end{align*}
We have
$\log(b_{m+1})=g_{m,1}(u_m) + \log(b_{m})$
where
\[g_{m,1}(u)= \log\left(1-\frac{2}{m(m+1)}+\frac{m+2}{m}
\left[\left(u-\frac{m}{m+1}\right)^+
+\left(\frac{1}{m+1}-u\right)^+\right]
\right)\]

The random sequence \((B_m)\), whose realization \((b_m)\) is generated recursively via \eqref{eq: uniform_updatingB}, converges almost surely 
by Theorem \ref{th: gen} since the two series 
\(\sum_{m=n} \E(g_{m,1}(U_m))\) and 
\(\sum_{m=n} \E(g_{m,1}(U_m)^2)\) converge. 
In order to verify this fact, 
consider that:
\begin{align*}
	g_{m,1}(u) &\leq 
	\log\left(1-\frac{2}{m(m+1)}+\frac{m+2}{m}
	\frac{1}{m+1}
	\left[ \ind_{[0,\frac{1}{m+1}]}(u)+
	\ind_{[\frac{m}{m+1},1]}(u) 
	\right]\right)\\
	g_{m,1}(u) &\geq 
	\log\left(1-\frac{2}{m(m+1)}\right)
\end{align*}
so that:
\begin{equation*}
	\begin{split}
		\E(g_{m,1}(U_m)) &\leq 
		\log\left(1-\frac{2}{m(m+1)}\right) \frac{m-1}{m+1}
		+ \log\left(1+\frac{1}{m+1}\right) \frac{2}{m+1} \\
		\E(g_{m,1}(U_m)) &\geq \log\left(1-\frac{2}{m(m+1)}\right)\\
		\E((g_{m,1}(U_m))^2) &\leq \left(\log\left(1-\frac{2}{m(m+1)}\right)\right)^2
		+ \left(\log\left(1+\frac{1}{m+1}\right)\right)^2 \frac{2}{m+1},
	\end{split}
\end{equation*}
which implies convergence of the two series 
$$\sum_{m=n} \E(g_{m,1}(U_m))\quad\mbox{and}\quad
\sum_{m=n} \E(g_{m,1}(U_m)^2)$$ since 
\(\lim_{x\to 0} \log(1+x)/x=1\).

At this stage, note that \(\E(W_m)=\E(V_m)=1/\{2(m+1)^2\}\) if $W_m=(1/(m+1)-U_m)^+$ and $V_m=(U_m-m/(m+1))^+$. 
Letting \(g_2(b_m, u_m)=b_m(v_m-w_m)\) we have that
\begin{align*}
	\E(g_2(B_m, U_m)\mid B_m)&=0,\\ 
	\E(g_2(B_m, U_m)^2 \mid B_m) &\leq B^2_m\E(V_m+W_m)
	=\frac{B^2_m}{(m+1)^2}
\end{align*}
and therefore the random sequence \((A_m)\), whose realization \((a_m)\) are generated recursively by \eqref{eq: uniform_updatingA}, converges almost surely by virtue of Theorem \ref{th: gen}.

\section{Nonlinear regression model}

Consider the following regression model,
$
Y= g(X, Z, \theta)
$
where \(g\) is a known regression function, \(Y\in \BR\) is the response, \(X\in\BR^p\) are the covariates, 
\(\theta\in\BR^d\) is the unknown parameter, \(Z\in\BR\) is the error following a known distribution and independent of \(X\). 

Given observed data \((x_1,y_1),\dotsc,(x_n,y_n)\), an estimate \(\hat{\theta}_n\) of \(\theta\) is obtained, for example by minimizing 
\(\sum_{i=1}^n L(\theta,x_i,y_i)/n\), where 
\begin{equation}\label{eq: L}
L(\theta,x,y)=(y-\mu(x,\theta))^2
\end{equation}
and 
\(
\mu(x,\theta)=\E(g(x, Z, \theta)).
\)
In other words,
\[
\hat{\theta}_n= \arg\min_{\theta\in\BR^d} \frac{1}{n}\sum_{i=1}^n
(y_i-\mu(x_i,\theta))^2
\]
We model the sequence of future observations as an asymptotically exchangeable random sequence \((X_{n+1},Y_{n+1}),(X_{n+2},Y_{n+2}),\dotsc\) with the aim to collect posterior samples of the random vector \(T\) such that:
\begin{equation}\label{eq: T}
T=\arg\min_{\theta\in\BR^d} \int_{\BR^p\times \BR} (y-\mu(x,\theta))^2\, \nu_\infty(\d x,\d y)
\end{equation}
almost surely (provided that \(T\) exists and is almost surely unique), where \(\nu_\infty\) is the almost sure weak limit of the empirical measure \(\nu_m\), as \(m\) diverges to infinity,
$
\nu_m=m^{-1}\sum_{i=1}^m \delta_{(X_i,Y_i)}. 
$ 
From \eqref{eq: L}, for \(j=1,\dotsc,d\), 
\begin{equation*}
	L'_{\theta_j}(\theta, x, y) = 2(y-\mu(x,\theta))\mu'_{\theta_j}(x,\theta),
\end{equation*}
future observations are obtained through the following recursive update based on stochastic gradient descent:
\begin{align}
Y_{m+1}&= g(X_{m+1}, Z_{m+1}, T_m) \label{eq: Ym}\\
T_{m+1,j}&= T_{m,j}- \frac{(Y_{m+1}-\mu(X_{m+1},T_{m}))\mu'_{\theta_j}(X_{m+1},T_m)}{
(m+1)\varphi(X_{m+1}, T_m)}
\label{eq: Tm}
\end{align}
for \(m=n,n+1,\dotsc\), where:
\[
\varphi(x,\theta)= 
\sqrt{\mu_2(x,\theta)-\mu(x,\theta)^2}
\max_{j=1,\dotsc,d}\lvert \mu'_{\theta_j}(x, \theta) \rvert,
\]
\(
\mu_2(x,\theta)=\E(g^2(x, Z, \theta)), 
\)
\(Z_{n+1},Z_{n+2},\dotsc\) are i.i.d. and distributed as \(Z\), and \(X_{n+1},X_{n+2},\dotsc\) are obtained through the Bayesian bootstrap (starting from \(x_1,\dotsc,x_n\)), and independent of \(Z_{n+1},Z_{n+2},\dotsc\). 

\begin{theorem}
Assume that \(\mu(x,\theta)\) is invertible as a function of \(\theta\), for every \(x\in\BR^p\). 
If \eqref{eq: Ym} and \eqref{eq: Tm} hold then 
\(\lim_{m\to \infty}T_m=T\), almost surely, where \(T\) satisfies \eqref{eq: T}. 
Moreover, the random sequence \((X_{n+1},Y_{n+1}),(X_{n+2},Y_{n+2}),\dotsc\) is asymptotically exchangeable. 
\end{theorem}

\begin{proof}

For each \(j=1,\dotsc,d\), \((T_{n+1,j},T_{n+2,j},\dotsc)\) is a martingale with respect to the filtration \((\mathscr{G}_{n+1},\mathscr{G}_{n+2},\dotsc)\) where 
\[\mathscr{G}_m=\sigma(Z_{n+1},\dotsc,Z_m,Y_{n+1},\dotsc,Y_m,X_{n+1},\dotsc,X_{m+1})\] for \(m=n+1,n+2,\dotsc\)
Moreover, by \eqref{eq: Tm}, we have that:
\begin{equation*}\begin{split}
\var(T_{m+1,j})&= \var(\E(T_{m+1,j}\mid \mathscr{G}_m))+
\E(\var(T_{m+1,j}\mid \mathscr{G}_m))\\
&\leq \var(T_{m,j})+\frac{1}{(m+1)^2}
\end{split}
\end{equation*}
so that for every \(m=n+1,n+2,\dotsc\), and \(j=1,\dotsc,d\),:
\[
\var(T_{m,j})\leq\sum_{\ell=n+1}^{m}\frac{1}{\ell^2}
< \sum_{\ell=n+1}^{\infty}\frac{1}{\ell^2}<\infty
\]
and therefore there exists a random variable \(T_j\) such that \(\lim_{m\to \infty}T_{m,j}=T_j\), almost surely, for \(j=1,\dotsc,d\). 

By Lemma 8.2 in \cite{Aldous85}, to prove asymptotic exchangeability we need to prove that there exists a random probability measure \(\nu\) on \(\BR^p\times \BR\) such that for every 
bounded and continuous function \(h_1:\BR^p\to \BR\) and \(h_2:\BR\to\BR\) such that:
\begin{equation}\label{eq: asympt-exch}
\lim_{m\to \infty}\E(h_1(X_{m+1})\,h_2(Y_{m+1})\mid X_{1:m},Y_{1:m})	
= \int_{\BR^p\times \BR} h_1(x)\,h_2(y)\, \nu(\d x,\, \d y)
\end{equation} 
To this aim, note that:
\begin{equation}\label{eq: step-1}
\begin{split}
\E(h_1(X_{m+1})\,&h_2(Y_{m+1})\mid X_{1:m},Y_{1:m})\\
&= \E(h_1(X_{m+1})\,\zeta(X_{m+1}, T_m)\mid X_{1:m},Y_{1:m})
\end{split}\end{equation}
where
$\zeta(x,\theta)=\E(h_2\circ g (x,Z,\theta))$.	
At this stage, recall that the sequence \(X_{n+1},X_{n+2},\dotsc\) obtained through the Bayesian bootstrap is exchangeable. Let \(X_0\) be a random variable independent of \(Z_{1:\infty}\) and of \(Y_{1:\infty}\) such that the sequence \((X_0,X_{n+1},X_{n+2},\dotsc)\) is exchangeable. In this way, the conditional distribution of \(X_0\) given \(X_{1:m}\) is the same of \(X_{m+1}\) given \(X_{1:m}\).  
Therefore,  \eqref{eq: step-1} becomes:
\begin{equation}\label{eq: step-2}
\E(h_1(X_{m+1})\,h_2(Y_{m+1})\mid X_{1:m},Y_{1:m})
= \E(h_1(X_{0})\,\zeta(X_{0}, T_m)\mid X_{1:m},Y_{1:m})
\end{equation}
At this stage, recall that \(\lim_{m\to \infty}T_m=T\), almost surely. 
By the dominated convergence theorem for conditional expectations \citep[see for instance Theorem 4.6.10 in][]{Durrett19}, \eqref{eq: step-2} yields:
\begin{equation*}
\lim_{m\to \infty}\E(h_1(X_{m+1})\,h_2(Y_{m+1})\mid X_{1:m},Y_{1:m})
= \E(h_1(X_{0})\,\zeta(X_{0}, T)\mid X_{1:\infty},Y_{1:\infty})
\end{equation*}
Note that \(T\) is measurable with respect to \(\sigma(X_{1:\infty},Y_{1:\infty})\), and the conditional distribution of \(X_0\) given \(X_{1:\infty},Y_{1:\infty}\) is almost surely the weak limit of the conditional distribution of \(X_{m+1}\) 
given \(X_{1:m}\) as \(m\) diverges to infinity and such limit is the de Finetti's measure of the sequence \(X_{n+1:\infty}\). Therefore, we obtain that 
\eqref{eq: asympt-exch} holds where \(\nu\) is the distribution of a random pair \((\tilde{X},\tilde{Y})\) such that
$\P(\tilde{X}=x^*_j)=P_j$, $j=1,\dotsc,k,$	
$\tilde{Y}=g(\tilde{X}, Z, T)$,
where 
\(x^*_1 \dotsc, x^*_k\) denote the distinct values of 
\(x_1,\dotsc,x_n\) in increasing order and \(n_j\) the number of times that the value \(x^*_j\) occurs in our data, for \(j=1,\dotsc,k\), 
\((P_1,\dotsc,P_{k-1})\) is a Dirichlet random vector with parameters \(n_1,\dotsc,n_k\) and \(P_k=1-\sum_{j=1}^k P_j\), and 
 \(Z\) is independent of 
\((\tilde{X},T)\). 

To conclude the proof we need to show that \(T\) satisfies \eqref{eq: T}, 
namely:
\begin{equation}\label{eq: optim}
T=\arg\min_{\theta\in\BR^d} \E(\tilde{Y}-\mu(\tilde{X},\theta))^2. 
\end{equation} 
If the random vector \(\tilde{T}\) is a solution of \eqref{eq: optim}, we have that
\[\begin{split}
\E(\tilde{Y}-\mu(\tilde{X},\tilde{T}))^2 &= 
\E( \E(\tilde{Y}-\mu(\tilde{X},\tilde{T})^2\mid \tilde{X}))\\
&= \E(\var(\tilde{Y}\mid \tilde{X}))+
\E(\mu(\tilde{X},T)-\mu(\tilde{X},\tilde{T}))^2 
\end{split}\]
so that \(\mu(\tilde{X},T)=\mu(\tilde{X},\tilde{T})\), almost surely. 
Being \(\mu(x,\theta)\)  invertible as a function of \(\theta\) for every \(x\in\BR^p\), we have that \(T=\tilde{T}\) almost surely and the proof is complete. 

\end{proof}

\subsection{Logistic regression}

To illustrate the above family of models, we consider a logistic regression model, i.e.
$$P(Y=1\mid x)=\frac{e^{x'\beta}}{1+e^{x'\beta}}$$
where $x$ is a $p$ dimensional vector of covariates and  $\beta$ a $p$ dimensional vector of parameters. The function
$g(X,Z,T)$ can be represented as $g(x,Z,\beta)=1(Z<P(Y=1\mid x))$ where $Z$ is an independent uniform random variable on $(0,1)$ and $\mu(x,\beta)=P(Y=1\mid x)$. It is easy to check that
$$\mu_j'(x,\beta)=\frac{x_j\,e^{x'\beta}}{(1+e^{x'\beta})^2}$$
and
$$\phi(x,\beta)=\sqrt{\mu(x,\beta)(1-\mu(x,\beta))}\,\frac{e^{x'\beta}}{(1+e^{x'\beta})^2}\,\max_j|x_j|.$$
Hence, one we have an estimator based on an observed sample of size $n$, say $\beta_n$, we can implement the algorithm generated by (\ref{eq: Ym}) and (\ref{eq: Tm}) to generate samples from the fiducial posterior distribution. 

\begin{center}
	\begin{figure}[!htbp]
		\begin{center}
			\includegraphics[width=14cm,height=16cm]{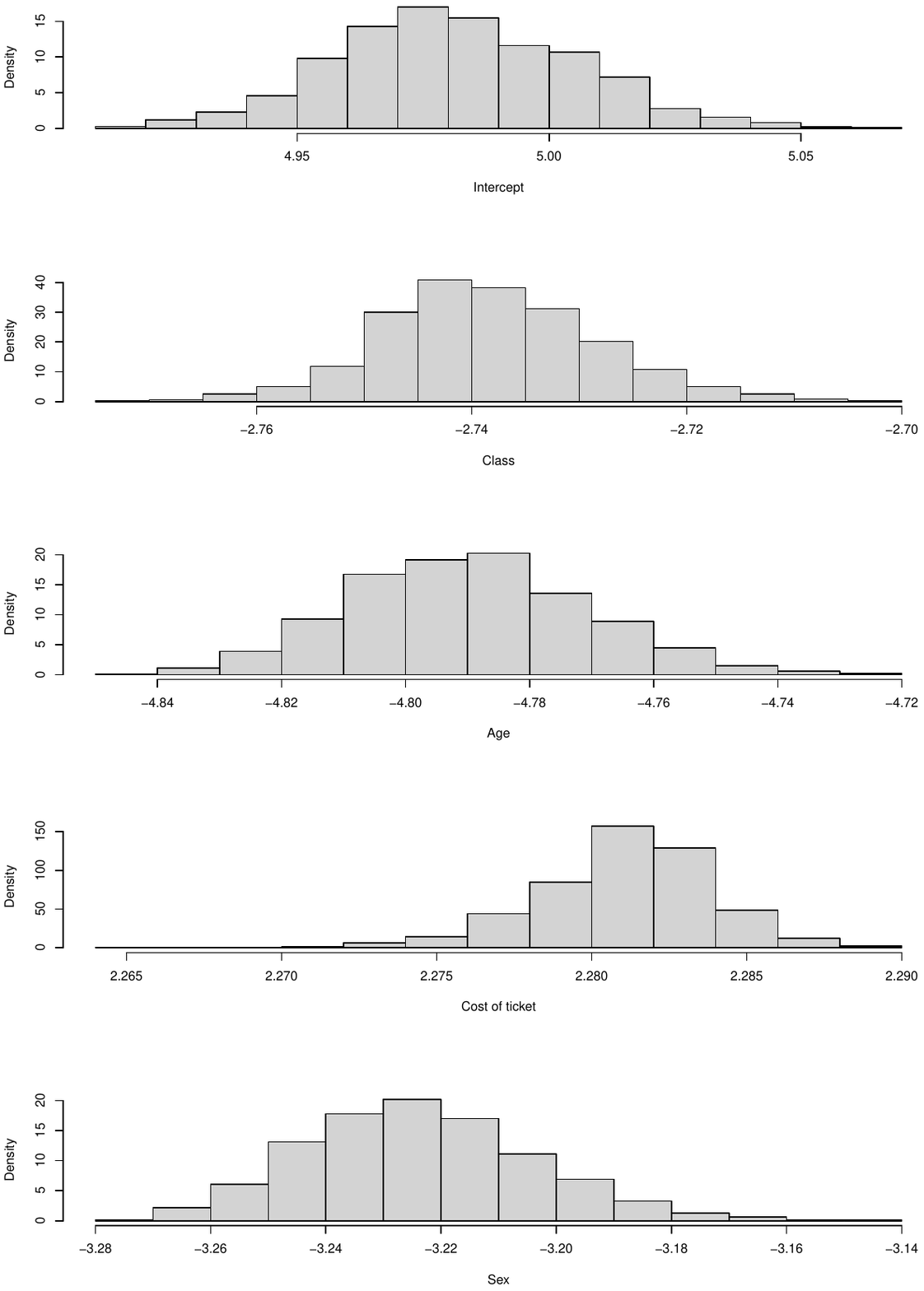}
			\caption{Fiducial posterior for regression parameters for the Titanic dataset}
			\label{figl}
		\end{center}
	\end{figure}
\end{center} 

We illustrate with the Titanic data set, which is available on CRAN, for example. The data consists of $n=887$ individuals and records their survival or otherwise along with covariate information including age, sex (male=1), cost of ticket, and class of travel, of which there are three types, 1,2 and 3. The covariates are all non-negative and standardized to lie between 0 and 1. We also include an intercept term, and so there are 5 unknown parameter values, which are estimated as 
$\widehat\beta=(4.98,-2.74,-4.79,2.28,-3.22).$
We then ran the algorithm generated by (\ref{eq: Ym}) and (\ref{eq: Tm}) for a length of size $n+1000$ with the starting point being $n+1$ and starting with $\beta_n=\widehat\beta$. 
A repetition of this over a 1000 runs leads to fiducial posterior samples from which we can construct histograms; see Fig.~\ref{figl}.

\section{Discussion}

Our approach to fiducial inference is general and based on the idea that we can generate a sequence of statistics, say $(T_m)_{m>n}$, where $T_n=t_n$ is observed. It is desirable to have
$T_{m+1}=H_m(T_m,z_m)$
for some sequence of functions $(H_m)$, where the $(Z_m)$ are a sequence of independent variables and independent of the $(T_m)$. There are many cases when this holds. For example, if $X_{m+1}=G(T_m,Z_m)$ and $T_{m+1}=G_m(T_m,X_{m+1})$.

The following procedure could always be used.
Consider the model in \cite{Fong2023} where the focus is on the distribution function and updated via
$$F_{m+1}(x)=(1-a_m)F_m(x)+a_m\,\Phi\left(\frac{\Phi^{-1}(F_m(x))-\rho \Phi^{-1}(F_m(X_{m+1}))}{\sqrt{1-\rho^2}}  \right),$$
where $X_{m+1}\sim F_m$. Hence, we can write
$$F_{m+1}(x)=(1-a_m)F_m(x)+a_m\,\Phi\left(\frac{\Phi^{-1}(F_m(x))-\rho\,Z_{m}}{\sqrt{1-\rho^2}}  \right),$$
where the $(Z_m)_{m>n}$ are independent standard normal random variables. Hence, in this case, we can write
$F_{m+1}=H_m(F_m,Z_m),\quad m\geq n.$
The limit $F_\infty$ is a sample from the Doob fiducial distribution. 
To adjust to a statistic $T_m$ it would be typical that
$T_m=\int g(x)\,d F_m(x)$
in which case, grouping $S_m=(T_m,F_m)$, we have 
$S_{m+1}=H_m(S_m,z_m).$


\end{document}